\documentclass{article16}
\usepackage[]{algorithm2e}

\title{A Linear-time Algorithm for Orthogonal Watchman Route Problem with Minimum Bends}
\author{Hamid Hoorfar \thanks{Department of Computer Engineering and Information Technology, Amirkabir University of Technology (Tehran Polytechnic), {\tt \{hoorfar,ar.bagheri\}@aut.ac.ir}}\and Alireza Bagheri\footnotemark[1]~\footnote{\textit{Coresponding author}.}
 }

\begin{document}
\maketitle

\begin{abstract}
	Given an orthogonal polygon $ P $ with $ n $ vertices, the goal of the watchman route problem is finding a path $ S $ of the minimum length in $ P $ such that every point of the polygon $ P $ is visible from at least one of the point of $ S $. In the other words, in the watchman route problem we must compute a shortest watchman route inside a simple polygon of $ n $ vertices such that all the points interior to the polygon and on its boundary are visible to at least one point on the route. If route and polygon be orthogonal, it is called orthogonal watchman route problem. One of the targets of this problem is finding the orthogonal path with the minimum number of bends as possible. We present a linear-time algorithm for the orthogonal watchman route problem, in which the given polygon is monotone. Our algorithm can be used also for the problem on simple orthogonal polygons $ P $ for which the dual graph induced by the vertical decomposition of $ P $ is a path, which is called path polygon. 
\end{abstract}

\section{Introduction}
The Watchman Route Problem is an optimization problem in computational geometry where the objective is to compute the shortest route that a watchman should take, to guard an entire area, given only the layout of the area. The layout of the area is represented usually by a simple polygon and the objective is to find a shortest closed curve such that all the points within the polygon and on its boundary are visible to at least one point of the curve~\cite{sagheer2017watchman}. The problem is defined in two different variants. In the fixed variant, the watchman should passes through a given boundary point. An analogy of it can be a situation in which the watchman or robot has to enter a building through a door (starting point) and guard the interior of it. In float variant, there is not any starting point for watchman (robot). Chin and Ntafos were the first worked out an algorithm that finds the shortest fixed watchman route~\cite{chin1991shortest} of $ O(n^{4}) $ time complexity. Later, Xuehou Tan proposed a solution of $ O(n^{3}) $ time
complexity to this version of the problem that constructed the shortest route incrementally~\cite{tan1993incremental}.The float variant of problem without any stating point, that constructed the shortest watchman route in worst case $ O(n^{6}) $ time~\cite{carlsson1999finding} and also, he himself had already submitted another solution to the problem for the fixed case using a divide and conquer approach. Later, Tan contributed to the floating case of the problem too and gave a solution of $ O(n^{5}) $ time complexity~\cite{sagheer2017watchman}, thus improving the previous time bound of $ O(n^{6}) $. Tan has given a linear time 2-approximation algorithm for the fixed watchman route~\cite{tan2007linear}. In many applied examples of the problem, we use a robot as the watchman who guards the entire area and in the many cases the robots have several motion restrictions because of their structures. For example, inexpensive robots have been built that only have the ability to move in two directions perpendicular to each other. These robots are faster and cheaper than ordinary robots and their use is very cost effective. Therefore, in many cases the cheapest route is orthogonal, absolutely, with the minimum number of bends. More bends will cost more. In the orthogonal variant of the problem, the area is an orthogonal simple polygon and the shortest route must be also an orthogonal polygonal curve with the minimum number of bends. According to our knowledge, it is the first time this problem has been introduced with this condition. Our algorithm gives a linear time solution for floating watchman route problem where the polygon is monotone and orthogonal and the watchman only is allowed to walk on an orthogonal route. We extend this result to the orthogonal \textit{path } polygons~\cite{de2014guarding}. Path polygons are orthogonal polygons for which the dual graph induced by the vertical decomposition of $ P $ is a path. This is the exact linear-time algorithm on $ x $-monotone orthogonal polygons that is used \textit{vertical decomposition}. The vertical decomposition of an orthogonal polygon $ P $ is the decomposition of $ P $ into rectangles obtained by extending the vertical edges incident to every reflex vertex of $ P $ inward until it hits the boundary of $ P $. The dual graph of the vertical decomposition is the graph that has a node for each rectangle in the decomposition and there is an edge between two nodes if and only if their corresponding rectangles are adjacent~\cite{de2014guarding}. A \textit{fan} polygon $ F $ is a simple polygon that has a vertex $ \nu $ which the entire polygon is visible to it. The vertex $ \nu $ is named \textit{core} vertex. Every fan polygon like $ F $ has a kernel and core vertex $ \nu $ belongs to the kernel, so, if a guard is placed on the core vertex, the polygon is covered, completely. See figure~\ref{fi:fig2}(b).

\section{Preliminaries}
Assume we decompose a simple orthogonal and $ x $-monotone polygon $ P $ with $ n $ vertices into rectangles obtained by extending the vertical edges incident to the reflex vertices of $ P $. A reflex vertex has an interior angle $ \frac{3\pi}{2} $ while convex vertices have an interior angle of $ \frac{\pi}{2} $. It is clear that every orthogonal polygon with $ n $ vertices has $ \frac{n-4}{2} $ reflex vertices. So, after the decomposition of $ P $, $ \frac{n-2}{2} $ rectangles will be obtained, exactly. Let $R =\{R_{1},R_{2},\dots,R_{m}\} $, where $ m = \frac{n-2}{2} $, be the set of rectangles, ordered from left to right according to $ x $-coordinate of their left edges. In the other words, after the decomposition, we named rectangular parts $ R_{1},R_{2},\dots,R_{m} $ from left to right. It is shown in Figure~\ref{fi:fig1}. We name the upper and lower horizontal edges of $ R_{i} $ by $ u_{i} $ and $ l_{i} $, respectively. Assume that $  U =\{u_{1},u_{2},\dots,u_{m} \}  $ and $ L=\{l_{1},l_{2},\dots,l_{m} \} $, {\small $1\leq i \leq m $}.
\begin{figure}
	\centering
	\includegraphics[width=0.8\textwidth]{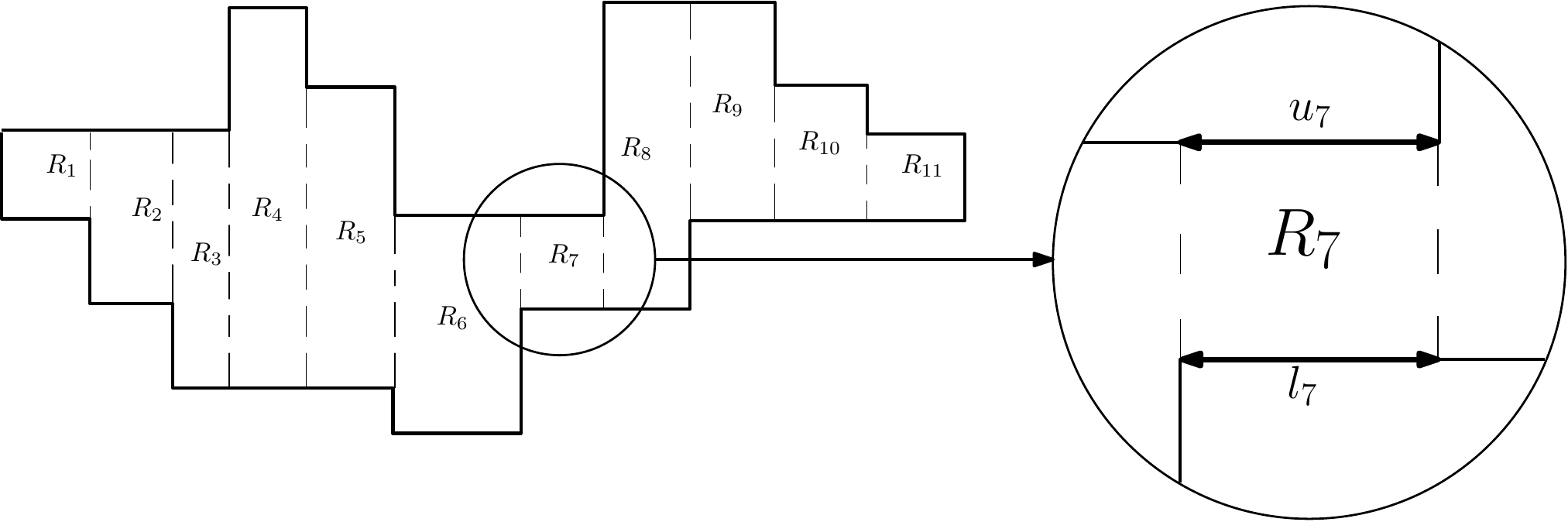}
	\caption{Illumination of the vertical decomposition and its notations.}
	\label{fi:fig1}
\end{figure}
For a horizontal line segment $ s $, we denote the $ x $-coordinate of the left vertex of $ s $ by $ x(s) $ that calls the $ x $-coordinate of line segment $ s $. Also, we denote the $ y $-coordinate of line segment $ s $ by $ y(s) $. It seems obvious that for all {\small $ 1 \leq i \leq m-1 $}, $ y(u_{i})=y(u_{i+1}) $ or $ y(l_{i})=y(l_{i+1}) $. So, we denote the edge of $ P $ that contains $ u_{i} $  by $ e(u_{i}) $ and the edge of $ P $ that contains $ l_{i} $  by $ e(l_{i}) $. Let $ E_{U}=\{e(u_{i})|1 \leq i \leq m\} $ and $ E_{L}=\{e(l_{i})|1 \leq i \leq m\} $, the sets of edges ordered from left to right. In a set $ E $ of horizontal edges of $ P $, $ e_{j} $ is named \textit{local maximum} if $ y(e_{j})\geq y(e_{j-1}) $ and $ y(e_{j})\geq y(e_{j+1}) $ and similarly, $ e_{k} $ is named \textit{local minimum} if $ y(e_{k})\leq y(e_{k-1}) $ and $ y(e_{k})\leq y(e_{k+1}) $. $ u_{m}$ (or $ l_{m}$) is named \textit{local maximum} if $ e( u_{m}) $ (or $ e(l_{m}) $) is local maximum, and $ u_{n}$ (or $ l_{n} $) is named \textit{local minimum} if  $ e( u_{n}) $ (or $ e(l_{n} )$) is local minimum. Therefore, in the set $ R $, $ R_{l} $ is named \textit{local maximum} if $ u_{l} $ and $ l_{l} $ are local maximum and local minimum, respectively. Two axis-parallel segments $ l $ and $ l' $ are defined as \textit{weak visible} if an axis-parallel line segment could be drawn from some point of $ l $ to some point of $ l' $ that does not intersect $ P $. If polygon $ P $ is $ x $-monotone and also $ y $-monotone, then $ P $ is named \textit{orthoconvex} polygon.
 \begin{lemma}
 \label{le:lemma001}
   For any orthoconvex polygon $ P $, if there exists a horizontal line segment $ \sigma_{1} $  which is connecting the leftmost and the rightmost vertical edges of $ P $ such that $ \sigma_{1}\in P $ and there exists a vertical line segment $ \sigma_{2} $  which is connecting the upper and the lower horizontal edges of $ P $ such that $ \sigma_{2}\in P $, then $ P $ has a kernel. If guard $ g $ occurs in the kernel, every point in $ P $ is guarded by it. See figure~\ref{fi:fig2}.
\end{lemma}
\begin{proof}
If $ \sigma_{1} $ connects the leftmost and the rightmost vertical edges of $ P $ and $ \sigma_{2} $  connects the upper and the lower horizontal edges of $ P $ then they have an intersection $\phi$ that is contained in $ P $. $\sigma_{1} $ and $ \sigma_{2} $ decompose $P$ into 4 sub-polygons $ P_{1}$,$ P_{2}$,$ P_{3}$ and $ P_{4}$. All of obtained sub-polygons are orthogonal fan polygons and $\phi$ is their core vertex, jointly. In every part, the entire sub-polygon is visible from  $ \phi $ and also it is on the kernel. Therefore, $ \phi $ belongs to the kernel of $ P $ and when guard $ g $ occurs in the kernel, every point in $ P $ is guarded by it.
\end{proof}
\begin{figure}
	\centering
	\includegraphics[width=0.7\textwidth]{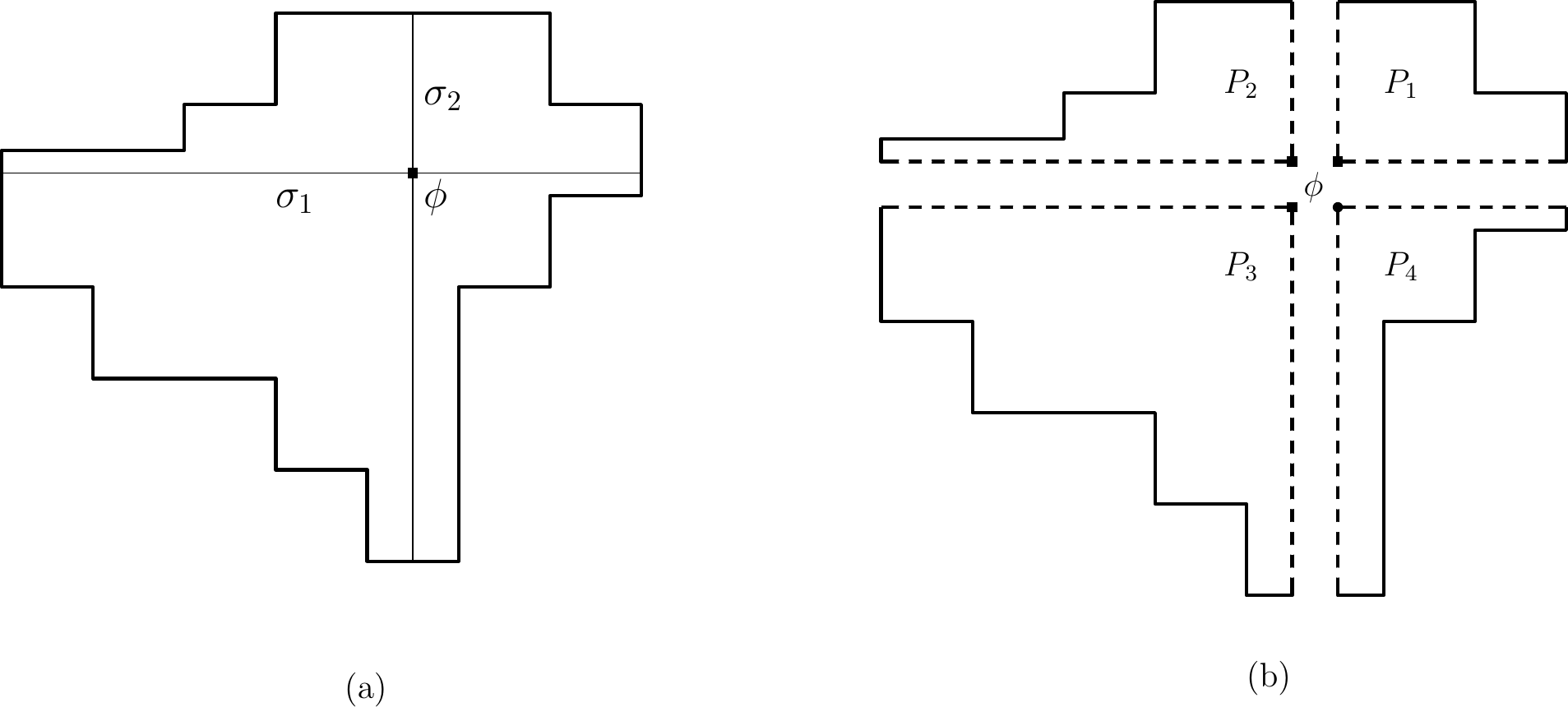}
	\caption{(a)An orthoconvex polygon $ P $ that has $ \sigma_{1} $ and $ \sigma_{2} $. Point $ \phi $ is the intersection of $ \sigma_{1} $ and $ \sigma_{2} $. (b)The decomposition of $ P $ into orthogonal fan polygons $ P_{1}$,$ P_{2}$,$ P_{3}$ and $ P_{4}$.}
	\label{fi:fig2}
\end{figure}
Given a monotone orthogonal polygon $ P $ with $ n $ vertices, if $ P $ be an orthoconvex, then floating point orthogonal watchman route problem has a simple result that is only a point on its kernel like that $ \phi $ in figure~\ref{fi:fig2}. So, it is an optimum result. In the next section, we provide an algorithm for finding the optimum route for every given monotone orthogonal polygon.
\section{The Algorithm of Orthogonal Watchman Route Problem for Guarding Monotone Areas}
Let $ P $ be an $ x $-monotone orthogonal polygon, we want to find an orthogonal route with minimum bends such that every point in $ P $ is visible along the route. In the other words, the route must be in such a way that the interior and boundary of $ P $ is weak visible from it. Therefore, we defined an $ x $-monotone orthogonal polygon whose optimum orthogonal route is a horizontal line segment and it is named \textit{balanced} polygon. An orthogonal polygon $ B $ is named balaced polygon if it is $ x $-monotone and there is a horizontal line segment that is connected the leftmost and right most vertical edges of the polygon with no intersection with other edges. We name this line segment as \textit{align}. The align line segment for the balanced polygon is not unique, in the every balanced polygon there is an orthogonal \textit{corridor} from leftmost edge to rightmost edges contains align segment. Every orthoconvex polygon is balanced, too. If an orthogonal polygon is not balanced it is named \textit{unbalanced}. In unbalanced polygon, there is no line segment (and corridor) that is connected the leftmost and rightmost edges without intersecting the other parts of the boundary. In the following, we will describe an algorithm to decompose an orthogonal $ x $-monotone polygon into balanced sub-polygons. The algorithm is based on finding corridors in the polygon. Let the leftmost vertical edge of $ P $ be $ \varepsilon $, for finding the first balanced part of $ P $, start from the leftmost vertical edge $ \varepsilon $ of $ P $, propagate a light beam in rectilinear path (straight-line path) perpendicular to $ \varepsilon $ and therefore collinear with the $ X $ axis. Remember that in vertical decomposition of $ P $, the set of obtained rectangles is notated by $ R $. All or part of this light beam passes through some rectangles of set $ R $. we name this subset of $ R $ as $ R_{\rho} $. These rectangles together make a sub-polygon $ \rho $ of $ P $ and it is the first balanced part of polygon. If $ R_{i} $ be a rectangle part of sub-polygon $ \rho $ and $ u_{i}$ and $l_{i} $ be the upper and lower horizontal edges of $ R_{i} $, respectively, for every $ u_{i}$ and $l_{i} $ belongs to $ \rho $, it is established that $ \min_{u_{i} \in \rho} (y(u_{i}))\geq \max_{l_{j} \in \rho}(y(l_{j})) $. So, there exists a horizontal line segment $ \sigma $  which is connecting the leftmost and rightmost vertical edges of $ \rho $ such that $ \sigma\in \rho $ because $ \rho $ is balanced. If we want to find an optimum orthogonal route for a balanced $ x $-monotone polygon with minimum length, it is possible to use align for finding the route. See Figure~\ref{fi:fig3} for an illustration. Now, If we remove $ \rho $ from $ P $ and iterate the described actions, $ P $ is decomposed to several balanced $ x $-monotone polygon. Now, we describe a linear-time algorithm for decomposition $ P $ to the balanced sub-polygons in algorithm~\ref{al:algo1}.
\begin{algorithm}[]
	\KwData{an $ x $-monotone polygon with $ n $ vertices}
	\KwResult{the shortest orthogonal route with minmum number of bends}
	(1)Set $ min_u=u_1$ and $ max_l=l_1 $\;
	(2)\ForEach{rectangle $ R_{i} $ belongs to $ R $}{
		(3)\If{ $ u_i>max_l $ or $ l_i<min_u $}{
			remove $ R_1,\dots,R_{i-1} $ from $ R $\;
			refresh the index of  $ R $ starting with $ 1 $\;
			go to 1\;
		}
		
		(4)Compute $min_u=\min(min_u,u_{i})$ and $max_l=\max(max_l,l_{i})$\;
}
\caption{Decomposition $ P $ into the balanced sub-polygons.}
\label{al:algo1}
\end{algorithm}
\begin{figure}
	\centering
	\includegraphics[width=\textwidth]{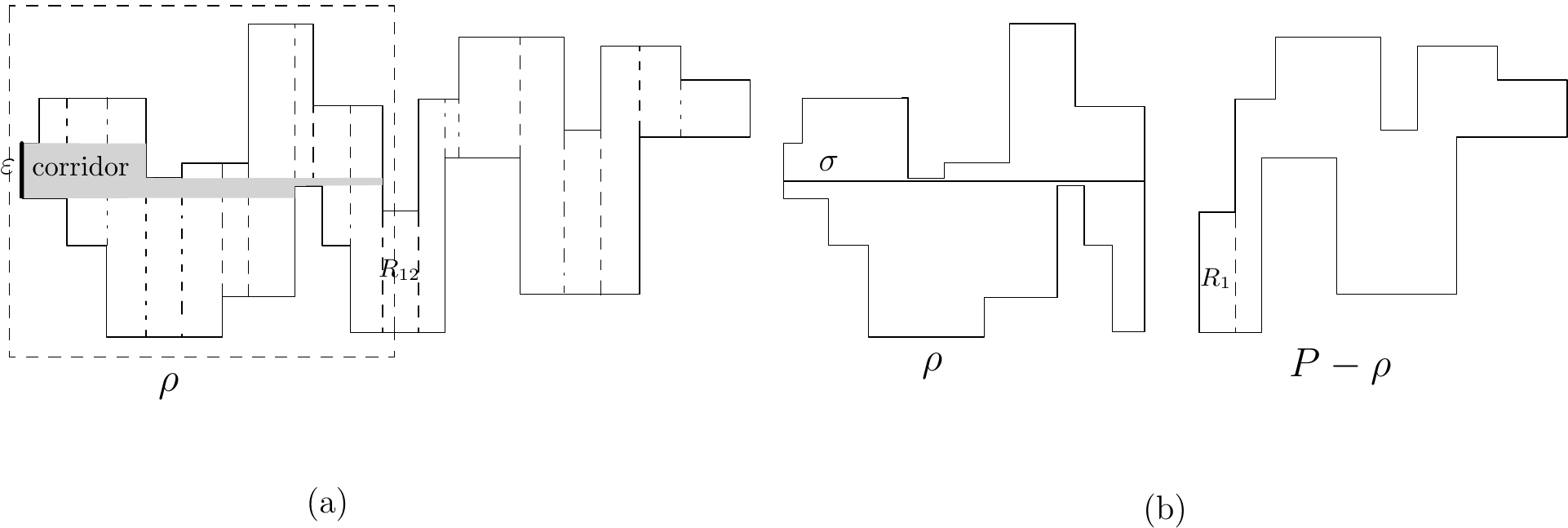}
	\caption{(a)Propagation a light beam to detemine sub-polygon $ \rho $. (b) Removing $ \rho $ from $ P $ and iterating the algorithm to obtain the decomposition.}
	\label{fi:fig3}
\end{figure}
\\If the condition in item 3 is satisfied, a balanced sub-polygon $ \rho $ is determined. So, we remove it from $ P $ and iterate algorithm for $ P-\rho $. We remove the rectangles belong to $ \rho $ from $ R $ as $ R=R-R_{\rho} $. We know the members of $ R $ are ordered from left to right and labeled from $ 1 $, after removing, we label the remained members from $ 1 $, again. Certainly, the same actions will be occurred for $ U $ and $ L $. The number of iterations is equal to the cardinality of $ R $ (in the start). Therefore, the time complexity for decomposition $ P $ into balanced polygons is linear. Every balanced polygon like $ \rho $ has a horizontal line-segment like $ \sigma $ which is connecting the leftmost and rightmost edges of $ \rho $ as named align. The entire $ \rho $ is visible from at least one point of align $ \sigma $, in the other words, $ \rho $ is weak visible from $ \sigma $. Therefore, $ \sigma $ is a candidate for being a route for the orthogonal watchman route problem in $ \rho $. Assume that $ P $ is decomposed into the balanced sub-polygons $ \rho_{1},\rho_{2},\dots,\rho_{k} $ and $ \sigma_{1},\sigma_{2},\dots,\sigma_{k} $ be their align line-segments, respectively. Paste $ \sigma_{1},\sigma_{2},\dots,\sigma_{k} $ together using $ k-1 $ vertical line-segments that connect the right end-point of $ \sigma_{i} $ to the left end-point of $ \sigma_{i+1} $, for every $ i $ between $ 1 $ and $ k-1 $. This obtained orthogonal path $ \varPi $ is the primary result for orthogonal floating watchman route problem in $ P $, if we trim some frills from its two ends. The watchman can ignore two parts of $ \varPi $ which are in the beginning and end of the route. The watchman can begin guarding from the leftmost point of the kernel of $  \rho_{1} $ that is intersected with $ \sigma_{1} $  to the rightmost point of kernel of $ \rho_{k} $ that is intersected with $  \sigma_{k} $. Now, we describe a linear-time algorithm for trimming  $ \varPi $ to the shortest orthogonal route as possible. The pseudo code of our algorithm is given in algorithm~\ref{al:algo2}. 
\begin{algorithm}[]
	\KwData{an orthogonal path $ \varPi $}
	\KwResult{the shortest orthogonal route}
	(1)\ForEach{rectangle $ R_{i} $ belongs to $ R $, from i=1 to m}{
	(2)	\If {$ u_{i}$ is local maximum or $ l_{i}$ is local minimum}{
			remove $ R_1,\dots,R_{i} $ from $ R $\;
			go to item 3\;
		}
	}

	(3)\ForEach{or each rectangle $ R_{i} $ belongs to $ R $, from i=m down to 1}{
		(4)\If{$ u_{i}$ is local maximum or $ l_{i}$ is local minimum}{
			remove $ R_{i},\dots,R_{m} $ from $ R $\;
			 go to 5\;
	}	}
		
	(5)Compute $ \varPi=\varPi\cap R $\;
	
\caption{Trimming path $ \varPi $ to the shortest orthogonal route}
\label{al:algo2}
\end{algorithm}
\\In the most of cases this algorithm is worked in constant-time but in the worst case it is order of $ O(n) $. This obtained route needs some secondary modification to become minimum orthogonal route as possible. This modification is about lengths of the vertical line segments that is connected every consecutive aligns. Let $\Sigma=\{ \sigma_{1},\sigma_{2},\dots,\sigma_{k}\} $ and $ y(\sigma_{i}) $ be $ y $-coordination of $ \sigma_{i} $ for every $ i $ between $ 1 $ and $ k $. Line segment $ \sigma_{i} $ is align segment for the sub-polygon $ \rho_i $, align segment for $ \rho_i $ is not unique and any horizontal line segment in the corridor of $ \rho_i $ can be its align segment. Some align segments in better than others and make the route shorter. The $ y $-coordination of align segments can be between two parameters in the balanced sub-polygon $ \rho_i $ and there are $ M_i=\min_{u_j\in R_{\rho_i}}y(u_j)$ and $ m_i=\max_{l_k\in R_{\rho_i}}y(l_k)$. We know that for every $ i $ between $ 1 $ and $ k $, $ M_{i} < m_{i+1}$ or $ m_i > M_{i+1} $. If $  M_{i-1}<M_i<M_{i+1}$ or $  M_{i-1}>M_i>M_{i+1}$, it does not matter whether align segment $ \sigma_i $ for $ \rho_i $ be such that $ y(\sigma_i)=M_i $ or $ y(\sigma_i)=m_i $, because the total length of the vertical line segments that are connecting two consecutive aligns $ \sigma_{i-1}, \sigma_{i} $ and $ \sigma_{i},\sigma_{i+1} $ is constant. But, if $  M_{i-1}<M_i>M_{i+1}$ or $  M_{i-1}>M_i<M_{i+1}$, this total length is not constant. In the former, we must select that align segment such that $ y(\sigma_i)=m_i $ and in the second case we must select that align segment that $ y(\sigma_i)=M_i $. If we do it, the route will be shorter because the total length of  $ \upsilon_{i-1} $ and $ \upsilon_{i} $ will be minimum where $\upsilon_{i-1} $ is the vertical line segment that connect right end point of $ \sigma_{i-1} $ and the left end point of $ \sigma_i $ and $ \upsilon_{i} $ is the vertical line segment that connect right end point of $ \sigma_{i} $ and the left end point of $ \sigma_{i+1} $. For the first balanced sub-polygon if $ M_1<M_2 $ then we must select align such that $ y(\sigma_1)=M_1 $ and if $ M_1>M_2 $ then we must select align such that $ y(\sigma_1)=m_1 $. For the last sub-polygon $ \rho_k $,if $ M_k<M_{k-1} $ then we must select align such that $ y(\sigma_k)=M_k $ and if $ M_k>M_{k-1} $ then we must select align such that $ y(\sigma_k)=m_k $. See figure~\ref{fi:fig4}.
 \begin{figure}
 	\centering
 	\includegraphics[width=\textwidth]{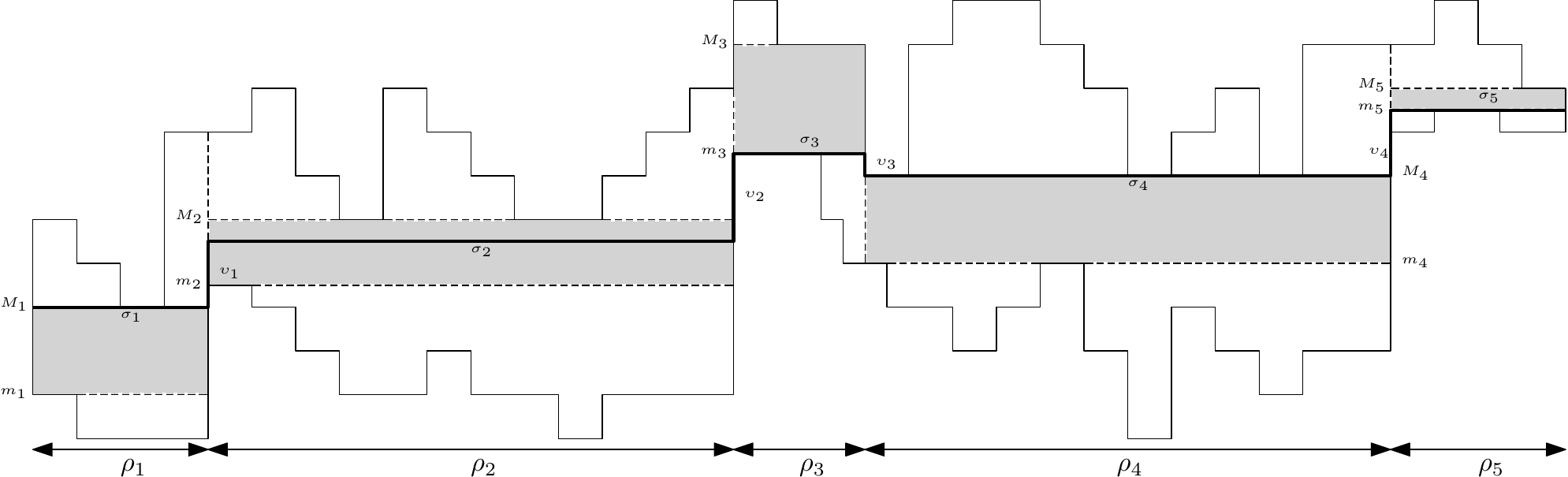}
 	\caption{An illustration of the selecting appropriate align segment for every balanced sub-polygon.}
 	\label{fi:fig4}
 \end{figure}
Now, we describe the linear-time algorithm for selecting appropriate align segments for every balanced sub-polygons to obtain the shortest orthogonal route as possible. The pseudo code of our algorithm is given in algorithm~\ref{al:algo3}. 
\begin{algorithm}[]
	\KwData{a set of available aligns}
	\KwResult{a set of appropriate aligns}
	(1)\ForEach{sub-polygon $ \rho_1 $ if $  M_1<M_2$}{
		set $ y(\sigma_1)=M_1 $, otherwise set $ y(\sigma_1)=m_1 $\;
	}
	
	(2)\ForEach{sub-polygon $ \rho_{i} $ belongs to $ P $, from $ i=1 $ to $ m $ }{
		(3)\eIf {$  M_{i-1}<M_i$ and $M_i>M_{i+1}$ then set $ y(\sigma_i)=m_i $}{
			set $ y(\sigma_i)=m_i $\;}{
			set $ y(\sigma_i)=M_i $\;
			}
		}
		
	(4)\ForEach{sub-polygon $ \rho_m $}{
		(5)\eIf {$  M_m<M_{m-1}$}{
					set $ y(\sigma_m)=M_m $\;}{
					set $ y(\sigma_m)=m_m$\;
					}
				}
				\{Variable $ m $ be the number of last sub-polygon.\}	
\caption{Selecting appropriate aligns}
\label{al:algo3}
\end{algorithm}
\\So far, we described an algorithm for finding the shortest orthogonal path with minimum number of bends. Our algorithm consists of three section, and we separately explained each of these three sections. To complete the algorithm, all three sections must be run sequentially. Next, we will explain the complexity of the algorithm in lemma~\ref{le:le2}.
\begin{lemma}
\label{le:le2}
The time complexity of described algorithm in the worst case is $ O(n) $.
\end{lemma}
\begin{proof}
In this section, we described our algorithm for finding the shortest orthogonal route with minimum bends for floating point watchman route problem in the linear time corresponding to the number of sides of the polygon. Our algorithm consists of three parts: 1.decomposition, 2.trimming and 3.selecting appropriate aligns. The time of all three parts are linear, so, the complete algorithm time is linear, too. This time is the best time complexity as possible because if we want to compute orthogonal route at least we must to check every vertices of the polygon. Therefore, our algorithm is tight and no algorithm better time is possible. The space complexity of the algorithm is also $ O(n) $.
\end{proof}

\section{Generalizing The Results}
Given an orthogonal polygon $ P $ with $ n $ vertices, after vertical decomposition, if the dual graph of decomposition be a path graph, then the polygon is named \textit{path} polygon~\cite{de2014guarding}. With this definition, every orthogonal monotone is path polygon. In many cases, the properties of path and monotone polygons are the same. we can use the algorithms that are explained in the previous section for the path polygons. In the other words, our algorithm can be solve the problem in linear time on any simple orthogonal polygon $ P $ for which the dual graph is a path. We can use the technique that is described in~\cite{de2014guarding} and convert path polygon $ P $ into a monotone by "unfolding". But in this section we use another strategy for finding the shortest orthogonal route for path polygons. See figure~\ref{fi:fig5}.
 \begin{figure}
 	\centering
 	\includegraphics[width=\textwidth]{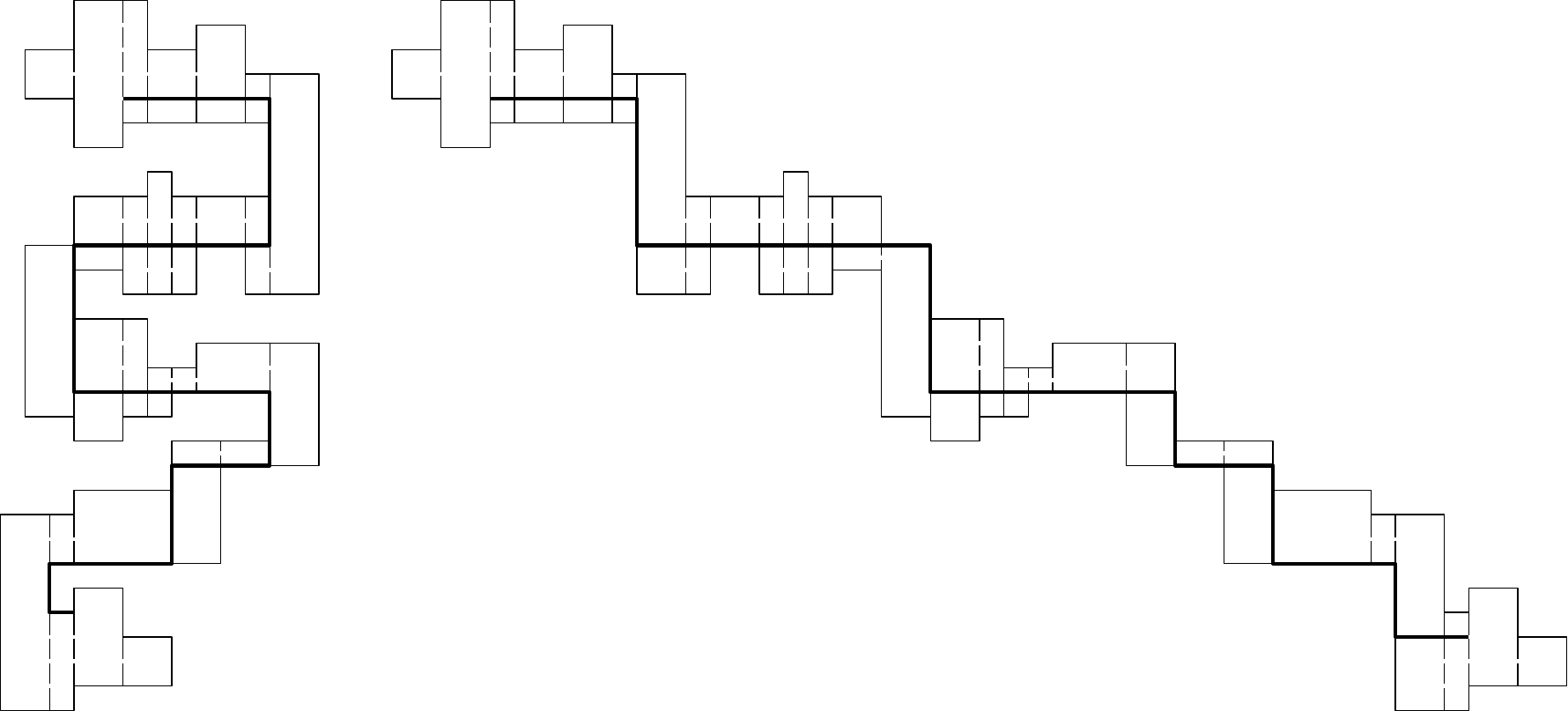}
 	\caption{Unfolding a path polygon and obtaining monotone polygon.}
 	\label{fi:fig5}
 \end{figure}
The difference between the path and monotone polygons is some rectangle parts which are named \textit{reflex rectangles}. A reflex rectangle is a rectangle part of vertical decomposition so that two neighbor parts are located in the same side of it, whether left side or right side. All other rectangle parts have two neighbor in the different sides of them. One in the left side and other one in the right side, except rectangle $ R_1 $ and $ R_m $ which are have only one neighbor part. In the first, we remove all of reflex rectangles from $ P $ by using this removal, the path polygon decompose into several monotone sub-polygons. the do our algorithm (decomposition into balanced) for every obtained monotone sub-polygon and do selecting appropriate aligns. connecting between routs that are obtained from balanced sub-polygons is simple. It is enough to draw a vertical line segment from their end points. After trimming, the obtained route is the shortest route as possible.

\section{Conclusion}
We considered a new version of floating point watchman route problem in which the shortest route must be orthogonal. We solved this problem in the linear time according to $ n $ where $ n $ is the number of sides of the given monotone polygon. the space complexity of our algorithm is $ O(n) $, too. After that we generalized our algorithm with the same time for path polygons. Both time and space complexity of our algorithm is $ O(n) $. It is the best for this new version of the problem. For the future works, we suggest to solve this problem for every simple orthogonal polygon with/without holes. 
\bibliographystyle{plain}
\bibliography{bibfile}

\begin{thebibliography}{1}

\bibitem{carlsson1999finding}
Svante Carlsson, H{\aa}kan Jonsson, and Bengt~J Nilsson.
\newblock Finding the shortest watchman route in a simple polygon.
\newblock {\em Discrete \& Computational Geometry}, 22(3):377--402, 1999.

\bibitem{chin1991shortest}
Wei-Pang Chin and Simeon Ntafos.
\newblock Shortest watchman routes in simple polygons.
\newblock {\em Discrete \& Computational Geometry}, 6(1):9--31, 1991.

\bibitem{de2014guarding}
Mark De~Berg, Stephane Durocher, and Saeed Mehrabi.
\newblock Guarding monotone art galleries with sliding cameras in linear time.
\newblock In {\em International Conference on Combinatorial Optimization and
  Applications}, pages 113--125. Springer, 2014.

\bibitem{sagheer2017watchman}
Muzammil Sagheer.
\newblock Watchman route problem.
\newblock {\em Technical Report}, 2017.

\bibitem{tan2007linear}
Xuehou Tan.
\newblock A linear-time 2-approximation algorithm for the watchman route
  problem for simple polygons.
\newblock {\em Theoretical Computer Science}, 384(1):92--103, 2007.

\bibitem{tan1993incremental}
Xuehou Tan, Tomio Hirata, and Yasuyoshi Inagaki.
\newblock An incremental algorithm for constructing shortest watchman routes.
\newblock {\em International Journal of Computational Geometry \&
  Applications}, 3(04):351--365, 1993.

\end{thebibliography}
\end{document}